%% file: main.tex
\documentclass[11pt]{article}
\usepackage{amsmath}
\usepackage{amsfonts}
\usepackage{amssymb}
\usepackage{amsthm}
\usepackage{thmtools, thm-restate}
\usepackage{fullpage}
\usepackage{comment}
\usepackage{hyperref}
\usepackage{lmodern}
\usepackage{caption}
\usepackage{lipsum}
\usepackage{graphicx}
\usepackage{subcaption}
\usepackage{appendix}
\usepackage[normalem]{ulem}

\usepackage[
    backend=biber,
    style=alphabetic,
    sorting=nyt,
    maxbibnames=99 % display all authors
]{biblatex}

\DeclareMathAlphabet{\mymathbb}{U}{bbold}{m}{n}

\declaretheorem[{style=definition,numberwithin=section}]{definition}
\declaretheorem[{style=definition,sibling=definition}]{theorem}
\declaretheorem[{style=definition,sibling=definition}]{lemma}
\declaretheorem[{style=definition,sibling=definition}]{claim}

\declaretheorem[{style=definition,sibling=definition}]{proposition}

\declaretheorem[{style=definition,sibling=definition}]{remark}

\DeclareMathOperator{\supp}{supp}

\newcommand{\weight}{\mathrm{wt}}
\newcommand{\dec}{\mathrm{Dec}}
\newcommand{\enc}{\mathrm{Enc}}

\newcommand{\basecode}{C_{\text{mother}}}

\newcommand{\typical}{\mathrm{Typical}}

\newtheorem*{theorem*}{Theorem}
\newtheoremstyle{named}{}{}{\itshape}{}{\bfseries}{.}{.5em}{\thmnote{#3}#1}
\theoremstyle{named}

\providecommand{\keywords}[1]
{
  \textbf{\textit{Keywords---}} #1
}

\title{
Capacity-Achieving Codes with Inverse-Ackermann-Depth Encoders
}
\author{
Yuan Li\footnote{Fudan University.  Email: \href{mailto:yuan_li@fudan.edu.cn}{yuan\_li@fudan.edu.cn}.}
}

\date{}

\begin{document}

\maketitle

\begin{abstract}
We prove that for any additive noise channel over $\mathbb{F}_q$, there exist error-correcting codes approaching channel capacity encodable by arithmetic circuits (with weighted addition gates) over $\mathbb{F}_q$ of size $O(n)$ and depth $2\alpha(n)$, where $\alpha(n)$ is a version of the inverse Ackermann function that is at most $3$ for all input lengths $n$ in practice.
Our results demonstrate that certain capacity-achieving codes admit highly efficient encoding circuits that are simultaneously of linear size and inverse-Ackermann depth. 
Our construction composes a linear code with constant rate and relative distance, based on the constructions of G\'{a}l, Hansen, Kouck\'{y}, Pudl\'{a}k, and Viola [IEEE Trans. Inform. Theory 59(10), 2013] and Drucker and Li [COCOON 2023], with an additional layer formed by a disperser graph.
A probabilistic argument over the edge weights of the 
disperser shows the existence of a deterministic encoder achieving error probability $2^{-\Omega(n)}$ at any rate below capacity.
\end{abstract}

\keywords{additive noise channel, error-correcting code, arithmetic circuit, inverse Ackermann function, disperser}

\input{Introduction}

\input{preliminaries}

\input{mother_code}

\input{additive_channel_proof}

\input{conclusion}

\printbibliography

\appendix

\end{document}

%% file: introduction.tex
\section{Introduction}

A fundamental problem in information theory and complexity theory is to understand the computational complexity of encoding and decoding error-correcting codes that approach channel capacity.

Shannon's noisy-channel coding theorem \cite{shannon1948mathematical} establishes the existence of a channel capacity $C$ for every discrete memoryless channel, such that for any rate below $C$, there exist encoding and decoding schemes whose error probability tends to zero.
The classical proofs of noisy-channel coding theorem are to analyze a random codebook, by an averaging argument, to show that there exist good codes that achieve channel capacity \cite{shannon1948mathematical, feinstein1954new, gallager1968information, mackay2003information, guruswami2012essential}. 
Indeed, such good codes are abundant: a random code (or even a random linear code) achieves channel capacity with high probability.

However, random codes have exceedingly high encoding complexity. In particular, a standard counting argument shows that random linear codes require encoding complexity $\Omega(n^2)$. 
This raises a natural question: \emph{How small can the encoding complexity be for capacity-achieving codes?}

Arithmetic circuits with unbounded fan-in provide a mathematically elegant computational model for measuring encoding complexity. By a simulation result \cite{stockmeyer1984simulation}, circuit depth and size --- measured 
by the number of wires --- correspond respectively to the parallel time and number of processors of a CRCW PRAM; thus, depth captures parallel running time and size captures total computational work. It is therefore natural to aim first for linear size and then to minimize depth.

G\'{a}l, Hansen, Kouck\'{y}, Pudl\'{a}k, and Viola \cite{GHK+12} thoroughly studied the circuit complexity of encoding codes with constant rate and constant relative distance over $\mathbb{F}_2$. For depth $2d$ and $2d+1$ circuit, their upper bound is $O_d(\lambda_d(n)\cdot n)$, matching their lower bound $\Omega_d(\lambda_d(n)\cdot n)$ for constant depth $d \ge 2$. Here, $\lambda_d(n)$ is a slowly growing function related to the inverse 
Ackermann function, e.g., $\lambda_1(n) = \lceil \log n \rceil$, $\lambda_2(n) = \log^* n$ and $\lambda_3(n) = \log^{**} n$.\footnote{Our 
$\lambda_d(n)$ corresponds to $\lambda_{2d}(n)$ in the notation 
of~\cite{GHK+12}.} Drucker and Li \cite{drucker2023minimum} improved their construction and analysis by tightening the upper bound to $O(\lambda_d(n)\cdot n)$, thereby removing the dependence on $d$. Letting $d = \alpha(n)$, where $\alpha(n)$ is a version of the inverse Ackermann function, one obtains a circuit of depth $2\alpha(n)$ and size $O(n)$; conversely, achieving linear size requires depth at least $2\alpha(n)-2$ \cite{GHK+12, drucker2023minimum}.

Several classical code families approach channel capacity.
Turbo codes admit linear-size but inherently sequential encoders \cite{berrou1993near}.
Polar \cite{Arikan09_polar} and Reed--Muller codes \cite{KLP12, ASW15, KKM+16, reeves2023reed, AS23} achieve capacity with $O(n\log n)$-size and $O(\log n)$-depth encoders,
while LDPC codes \cite{Luby1997_practicalldpc, RU01_LDPC_encode, LM09_linear_ldpc} allow linear-size encoding with logarithmic or larger depth.
None of these constructions achieve linear size together with inverse-Ackermann depth.

Our main result is the following:
\begin{theorem}
\label{thm:main}
Let $q$ be a prime power, and let $\Pi$ be an additive noise channel over $\mathbb{F}_q$, i.e., $p(z \mid x) = p_N(z - x)$ for some distribution 
$p_N$ over $\mathbb{F}_q$.
For any rate $r$ below the channel capacity $C(\Pi)$, and for sufficiently large block length $n$, there exists a linear error-correcting code with
\begin{itemize}
    \item encoder $\mathrm{Enc}: \mathbb{F}_q^k \to \mathbb{F}_q^n$ and decoder $\mathrm{Dec}: \mathbb{F}_q^n \to \mathbb{F}_q^k \cup \{\text{fail}\}$,
    \item rate $\frac{k}{n}\cdot \log q \ge r$ and worst-case error probability $2^{-\Omega_{\Pi, r}(n)}$,
    \item an encoder implementable by an arithmetic circuit over $\mathbb{F}_q$ consisting solely of weighted addition gates with unbounded fan-in, of depth $2\alpha(n)$ and size $O_{\Pi,r}(n)$.
\end{itemize}
\end{theorem}
\begin{remark}
The additive noise assumption $p(z \mid x) = p_N(z - x)$ is without 
significant loss of generality: it encompasses virtually all channels 
of practical interest over $\mathbb{F}_q$, including the $q$-ary 
symmetric channel and the binary symmetric channel (BSC).
\end{remark}
\begin{remark}
The function $\alpha(n)$ is a version of the inverse Ackermann function, which grows extremely slowly with $n$. For example, $\alpha\!\left(2^{65536}\right)=3$. Thus for all practical purposes, the encoder circuit has depth at most 6.
\end{remark}
\begin{remark}
Our construction is probabilistic: the edge coefficients of the disperser layer are chosen at random, making the construction resemble a random linear code, for which decoding is conjectured to be computationally hard.
Derandomizing the circuit graph and the coefficients remains an open problem;
it is possible that some code in this ensemble admits an efficient decoder.
\end{remark}

The classical proof of noisy-channel coding theorem reveals that capacity-achieving codes are abundant. Our result further demonstrates that certain capacity-achieving codes admit encoders that are also computationally inexpensive — they can be implemented by linear-size arithmetic circuits of inverse-Ackermann depth.

\textbf{Proof techniques.} Our proof relies on two ingredients: a mother code that can be encoded by a linear circuit in inverse Ackermann depth, and a disperser graph with random coefficients forming the final layer of the circuit.

The first ingredient is a mother code $\mathbb{F}_q^{k} \to \mathbb{F}_q^{32k}$ that can be encoded by a linear arithmetic circuit of depth $d$ and size $O_q(\lambda_d(n) \cdot n)$, based on the construction in \cite{GHK+12,drucker2023minimum}. We extend the construction from $\mathbb{F}_2$ to $\mathbb{F}_q$, which requires only routine verification.

The second ingredient is a disperser graph that forms the circuit’s final layer, where each vertex is replaced by a weighted addition gate and each edge is assigned a coefficient chosen independently and uniformly at random from $\mathbb{F}_q$.
Fed with the high-weight output of the mother code, the disperser layer 
randomizes at least $(1-\gamma)n$ output coordinates uniformly over 
$\mathbb{F}_q$, mimicking the behavior of a random linear code.
This construction, inspired by linear network coding \cite{li2003linear}, was used in prior work \cite{drucker2023minimum} to amplify the code rate up to the Gilbert--Varshamov bound (over $\mathbb{F}_2$).

\textbf{Related techniques.} Druk and Ishai~\cite{druk2014linear} gave a randomized construction achieving the GV bound over any finite field $\mathbb{F}_q$, encodable by linear-size arithmetic circuits with bounded fan-in. Their approach was to construct a \emph{linear uniform-output family} and set certain inputs independently and uniformly at random.

Graph-concatenated codes based on disperser graphs were introduced by Guruswami and Indyk \cite{guruswami2001expander}, with subsequent works in, e.g., \cite{guruswami2004better, rom2006improving, li2025improved}. In that framework, if the mother code is over alphabet $\mathbb{F}_q$, and the disperser graph has left degree $d$, the resulting concatenated code is over $\mathbb{F}_q^d$ (or $\mathbb{F}_{q^d}$). Our construction differs in two respects: it preserves the base alphabet $\mathbb{F}_q$, and it employs randomly chosen coefficients alongside the disperser graph.

\textbf{Organization.} Section 2 introduces the necessary definitions and notation. Section 3 establishes the existence of such a mother code over an arbitrary finite field. 
Section 4 proves our main result. Section 5 concludes the paper.

%% file: preliminaries.tex
\section{Preliminaries} \label{ch:prelim}

\subsection{Discrete Memoryless Channel}

A \emph{discrete memoryless channel} $\Pi$ consists of input and output alphabets $\mathcal{X}$, $\mathcal{Y}$, and transition probabilities $(p(y \mid x))_{x \in \mathcal{X},\, y \in \mathcal{Y}}$, with $p(y \mid x) = \prod_{i=1}^n p(y_i \mid x_i)$ for $x \in \mathcal{X}^n$, $y \in \mathcal{Y}^n$.

A channel is \emph{symmetric} if $|\mathcal{X}| = |\mathcal{Y}| = q$ and every row of its transition matrix is a permutation of every other row, and likewise for columns. Consequently, every row and every column contains the same multiset of probabilities.

A symmetric channel over $\mathbb{F}_q$ is an \emph{additive noise channel} 
if $\mathcal{X} = \mathcal{Y} = \mathbb{F}_q$ and $p(z \mid x) = p_N(z - x)$ 
for some distribution $p_N$ over $\mathbb{F}_q$, where arithmetic is in 
$\mathbb{F}_q$; equivalently, the channel output is $x + \eta$ where 
$\eta \sim p_N$ is independent of the input.

By Shannon's noisy-channel coding theorem,
$
C(\Pi) = \max_{p(x)} I(X;Y),
$
where $p(x)$ ranges over all distributions on $\mathcal{X}$. For symmetric channels, the maximum is achieved by the uniform input distribution, which induces a uniform output distribution, giving
\begin{equation}
\label{equ:sym_channel_capacity}
C(\Pi) = \log q - H_\Pi,
\end{equation}
where $H_\Pi$ denotes the binary entropy of the first row of the transition matrix.

\subsection{Error-Correcting Code}

Let $C : \mathbb{F}_q^{k} \to \mathbb{F}_q^{n}$. The \emph{rate} of $C$ is $k/n$. Its \emph{minimum distance} is
\[
d(C) = \min_{\substack{x, y \in \mathbb{F}_q^k \\ x \ne y}} \mathrm{dist}(C(x),C(y)),
\]
where $\mathrm{dist}(u,v) = |\{i \in [n] : u_i \ne v_i\}|$ denotes Hamming distance, and its \emph{relative distance} is $d(C)/n$.

Let $\enc:\mathbb{F}_q^{k} \to \mathbb{F}_q^{n}$ be an encoder and $\dec:\mathbb{F}_q^{n} \to \mathbb{F}_q^{k} \cup \{\text{fail}\}$ a decoder. The \emph{failure probability} of $(\enc, \dec)$ over a channel $\Pi$ is
\[
P_e = \max_{m \in \mathbb{F}_q^k} \sum_{y \in \mathbb{F}_q^n} p(y \mid \enc(m))\, \mathbf{1}[\dec(y) \neq m].
\]

\subsection{Arithmetic Circuit}

Fix a finite field $\mathbb{F}_q$. A \emph{linear circuit} over $\mathbb{F}_q$ consists of weighted addition gates of unbounded fan-in, where a gate with inputs $x_1,\dots,x_s$ computes $\sum_{i=1}^s c_i x_i$ for fixed $c_i \in \mathbb{F}_q$; such circuits compute linear functions over $\mathbb{F}_q$. The \emph{size} of a circuit is its total number of \emph{wires}, and its \emph{depth} is the length of its longest input-to-output path.

\subsection{Ackermann Function}

\begin{definition}
\label{def:acker}
(Ackermann function \cite{Tar75, DDP+83})
The \emph{Ackermann function} $A : \mathbb{N} \times \mathbb{N} \to \mathbb{N}$ is defined recursively by
\begin{equation}
A(i, j) \,=\,
\begin{cases}
    2j,                   & i = 0,\ j \ge 1,   \\
    2,                    & i \ge 1,\ j = 1,   \\
    A(i-1,\, A(i, j-1)), & i \ge 1,\ j \ge 2.
\end{cases}
\end{equation}
%We write $A_i(j)$ as shorthand for $A(i,j)$.
\end{definition}

For example, $A(0, j) = 2j$, $A(1, j) = 2^j$ and $A(2, j) = 2\uparrow j$, where $2\uparrow j$ denotes a tower of $j$ twos:
\begin{equation*}
2\uparrow j \;=\; \underbrace{2^{2^{\cdot^{\cdot^{\cdot^{2}}}}}}_{j}.
\end{equation*}

\begin{definition} \cite{RS03}
\label{def:iter}
For a function $f : \mathbb{N} \to \mathbb{N}$, the \emph{iterated function}
$f^* : \mathbb{N} \to \mathbb{N}$ is defined as
\begin{equation}
    f^*(n) \;=\; \min\bigl\{\, k \ge 0 : f^{(k)}(n) \le 1 \,\bigr\},
\end{equation}
where $f^{(0)}(n) = n$ and $f^{(k)}(n) = f\bigl(f^{(k-1)}(n)\bigr)$
for $k \ge 1$.
\end{definition}

\begin{definition}
\label{def:inv-acker} For each integer $d \ge 0$, the function $\lambda_d : \mathbb{N} \to \mathbb{N}$
is defined as
\begin{equation}
    \lambda_d(n) \;=\; \min\bigl\{\, m \ge 1 : A(d,\, m) \ge n \,\bigr\}.
\end{equation}
\end{definition}

For example, \[
\lambda_0(n)=\left\lceil \tfrac{n}{2} \right\rceil,\quad
\lambda_1(n)=\left\lceil \log_2 n \right\rceil,\quad
\lambda_2(n)=\log^* n.
\]

\begin{proposition}
\label{prop:lambda-closed-form}
For all $n \ge 1$,
\begin{align*}
    \lambda_0(n) &\;=\; \left\lceil \frac{n}{2} \right\rceil,  \\
    \lambda_d(n) &\;=\; \lambda^*_{d-1}(n), \qquad d \ge 1.   
\end{align*}
\end{proposition}
\begin{proof}
For $d=0$, since $A(0,m)=2m$, we have
\[
\lambda_0(n)=\min\{m\ge1:2m\ge n\}=\left\lceil \tfrac{n}{2}\right\rceil.
\]
For $d\ge1$, using $A(d,m)=A(d-1,A(d,m-1))$ and the definition of $\lambda_{d-1}$, we have
\[
A(d,m)\ge n
\;\Longleftrightarrow\;
A(d,m-1)\ge \lambda_{d-1}(n).
\]
Iterating this equivalence yields
\[
A(d,m)\ge n
\;\Longleftrightarrow\;
\lambda_{d-1}^{(m-1)}(n)\le 2,
\]
which is equivalent to $\lambda_{d-1}^{(m)}(n)\le 1$. Hence
\[
\lambda_d(n)=\min\{m\ge1:A(d,m)\ge n\}
=\min\{k\ge0:\lambda_{d-1}^{(k)}(n)\le1\}
=\lambda_{d-1}^*(n).
\]
\end{proof}

\begin{definition}
\label{def:alpha}
The \emph{inverse Ackermann function} $\alpha : \mathbb{N} \to \mathbb{N}$ is defined as
\begin{equation}
    \alpha(n) \;=\; \min\bigl\{\, d \ge 0 : \lambda_d(n) \le 3 \,\bigr\}.
\end{equation}
\end{definition}

The function $\alpha(n)$ is well-defined since $A(i, 3) \ge i$ for all $i \ge 1$, so $\lambda_n(n) \le 3$ for any $n \ge 1$, witnessing $d \le n$. The function $\alpha(n)$ grows extremely slowly: for example, $\alpha(n) \le 3$ for all $n \le 2^{65536}$.

\subsection{Dispersers and Superconcentrators}

\begin{definition} \cite{guruswami2012essential} A bipartite graph $G = (L = [n], R = [m], E)$ is a $(\gamma, \varepsilon)$-disperser if for all subsets $S \subseteq L$ with $|S| \ge \gamma n$, we have $|N(S)| \ge (1 - \varepsilon)m$.
\end{definition}

\begin{theorem}[Theorem 1.10 in \cite{RT00}, restated] \label{thm:prob_disperser} Let $c > 0$ and $\gamma, \varepsilon > 0$. For any positive integer $n$ and $m = \lfloor cn \rfloor$, there exists a $(\gamma, \varepsilon)$-disperser graph $G = (V_1 = [n], V_2 = [m], E)$ with degree bounded by $O_{c, \gamma, \varepsilon}(1)$.
\end{theorem}

The original Theorem 1.10 in \cite{RT00} is more general, allowing cases where $m \gg n$. We state a simplified version tailored to our needs. The original theorem guarantees that the left degree is bounded; the right degree is unbounded.
By applying a purging argument that discards the half of the right vertices with the highest degrees, one can also ensure that the right-degree remains bounded.

\begin{definition} \cite{Val75}
A directed acyclic graph on $m$ inputs and $n$ outputs is a \emph{superconcentrator} if for every $k \le \min\{m, n\}$ and every subset $S$ of $k$ inputs and $T$ of $k$ outputs, there exist $k$ vertex-disjoint paths from $S$ to $T$.
\end{definition}

%% file: mother_code.tex
\section{Mother Code Construction}

A linear map $C: \mathbb{F}_q^n \to \mathbb{F}_q^{32n}$ is a \emph{good code} if $\mathrm{wt}(C(x)) \ge 4n$ for every nonzero $x \in \mathbb{F}_q^n$ \cite{drucker2023minimum}, where the constant $32$ follows~\cite{GHK+12}.

A linear map $C : \mathbb{F}_q^{n} \to \mathbb{F}_q^{32n}$ is an \emph{$(n, r, s)$-partial good code}, $(n,r,s)$-PGC for short, if $\mathrm{wt}(C(x)) \ge 4n$ for every nonzero $x \in \mathbb{F}_q^{n}$ with $\mathrm{wt}(x) \in [r, s]$, where $\mathrm{wt}(x) = |\{i : x_i \neq 0\}|$ denotes the Hamming weight of $x$.

\begin{definition} \label{def:range_detector}
\cite{GHK+12} An \emph{$(m, n, \ell, k, r, s)$-range detector} is a map $C:\mathbb{F}_q^m \to \mathbb{F}_q^n$ such that $\mathrm{wt}(C(x)) \in [r, s]$ for every $x \in \mathbb{F}_q^m$ with $\mathrm{wt}(x) \in [\ell, k]$. When $s = n$, we omit the last parameter and write $(m, n, \ell, k, r)$-range detector.
\end{definition}

A good code is a special case of a range detector: an $(n, r, s)$-PGC is an $(n, 32n, r, s, 4n)$-range detector.

Fix a finite field $\mathbb{F}_q$. Let $S_d(n)$ denote the minimum size of a depth-$d$ linear circuit over $\mathbb{F}_q$ computing a good code $C : \mathbb{F}_q^n \to \mathbb{F}_q^{32n}$. More generally, let $S_d(n, r, s)$ denote the minimum size of a depth-$d$ linear circuit over $\mathbb{F}_q$ computing an $(n,r,s)$-PGC $C : \mathbb{F}_q^n \to \mathbb{F}_q^{32n}$.

Our goal is to prove the following theorem. The proof closely follows that of Theorem~1.1 in \cite{drucker2023minimum}, which refines the construction and analysis of \cite{GHK+12}; we therefore present only a proof sketch, highlighting the differences.

\begin{theorem}
\label{thm:mother_code}
Fix a finite field $\mathbb{F}_q$. For any positive integer $n$ and integer $d \ge 0$,
\[
S_{2d}(n) = O_q\!\bigl(\lambda_d(n) \cdot n\bigr).
\]
In particular, taking $d = \alpha(n)$ gives $S_{2\alpha(n)}(n) = O_q(n)$.
\end{theorem}

\subsection{Construction Overview}

We begin by outlining the construction and its main building blocks. In its structure, the entire construction closely resembles superconcentrators \cite{Val75, Val76, Val77, DDP+83}.

Valiant introduced superconcentrators while studying circuit lower bounds. For proving circuit lower bounds, the existence of ultra-low depth, linear-size superconcentrators is a negative result, as it rules out superlinear 
lower bounds based solely on information-transfer arguments. On the other hand, superconcentrators eliminate information-transfer bottlenecks, making them useful in computation and communication tasks, for example, secret sharing \cite{Li23}.

The construction relies on the following key components:
\begin{itemize}
    \item \textbf{Distance amplifier.} A \emph{distance amplifier} can improve the rate and/or the relative distance from any constant to near the Gilbert--Varshamov bound, and it can be computed by a depth-1 circuit of size $O(n)$ with bounded fan-in (Lemma~\ref{lem:rate_booster}).

    \item \textbf{Output amplifier.} An \emph{output amplifier} arbitrarily increases the number of output coordinates while maintaining a constant relative Hamming weight, realizable by a depth-$1$ circuit of linear size (Lemma~\ref{lem:output_amplifier}).
    
    \item \textbf{Condenser.} A \emph{condenser} reduces the input length 
    from $n$ to $n/r$ while preserving a lower bound on the output Hamming 
    weight, computable by a depth-$1$ circuit of linear size 
    (Lemma~\ref{lem:condenser}).
    
    \item \textbf{Composition Lemma.} Combines several PGCs into a larger one, increasing the depth by 1 while keeping the output fan-in bounded by a constant (Lemma~\ref{lem:composition}).
\end{itemize}

%Using these building blocks, a good code can be constructed recursively with inverse-Ackermann depth and a linear number of wires.

\subsection{Proof of Theorem \ref{thm:mother_code}}

\begin{lemma} [Distance Amplifier]
\label{lem:rate_booster}	
Let $C:\mathbb{F}_q^n \to \mathbb{F}_q^{32n}$ be a code of relative distance $\rho > 0$.
For any $c > 1$ and $\delta > 0$ satisfying $\tfrac{1}{c} < 1 - H_q(\delta)$, there exists a
linear map
\[
L : \mathbb{F}_q^{32n} \to \mathbb{F}_q^{\lfloor cn \rfloor}
\]
such that $L(C(x))$ has relative distance at least $\delta$. Moreover, $L$ is computable by
depth-1 linear circuits of size $O_{\rho,c,\delta}(n)$, and all output gates have bounded
fan-in $O_{\rho,c,\delta}(1)$.
\end{lemma}
\begin{proof}
Identical to that of Lemma~3.4 of \cite{drucker2023minimum}, generalized to $\mathbb{F}_q$.
We take a bounded-degree $(\delta,\varepsilon)$-disperser $G=([32n],[cn],E)$, replace each right vertex in $[cn]$ by a weighted addition gate, and assign to every edge an independently and uniformly chosen coefficient from $\mathbb{F}_q$.

Let $x \in \mathbb{F}_q^n$ be nonzero, and let $N(C(x)) \subseteq [\lfloor cn \rfloor]$ denote the set of output coordinates adjacent in $G$ to at least one nonzero coordinate of $C(x)$.
The distribution of $G(C(x))\restriction_{N(C(x))}$ is uniform in $\mathbb{F}_q^{N(C(x))}$, with randomness arising from the coefficients on the edges. Applying a union bound and using the inequality
\begin{equation}
\label{equ:bsum_fq}
\sum_{i=0}^{\gamma n} \binom{n}{i}(q-1)^i \le q^{H_q(\gamma)n},
\end{equation}
we conclude that there exists such a linear map $L$.
\end{proof}

The following lemma allows us to combine multiple PGCs into a larger one, which will be applied repeatedly in the construction.

\begin{lemma} [Composition Lemma]
\label{lem:composition}
For any $1 \le r_1 < \cdots < r_{t+1} \le n$, we have
\[
S_{d+1}(n, r_1, r_{t+1}) \le \sum_{i=1}^t S_d(n, r_i, r_{i+1}) + O(tn).
\]Furthermore, if, for each $i = 1, \ldots, t$, there exists an $(n, r_i, r_{i+1})$-PGC computable by a depth-$d$ size-$s_i$ linear circuit with output gates of bounded fan-in $D$, then
$
S_d(n, r_1, r_{t+1}) \le \sum_{i=1}^t s_i + O(D t n),
$
and the output gates of the combined circuit have bounded fan-in $O(Dt)$.
\end{lemma}
\begin{proof}
The proof is analogous to Lemma 3.5 in \cite{drucker2023minimum}, but over a general field $\mathbb{F}_q$. We provide a sketch of the argument.

Let $C_i:\mathbb{F}_q^n \to \mathbb{F}_q^{32n}$ denote an $(n, r_i, r_{i+1})$-PGC computable by a linear circuit of size $S_d(n, r_i, r_{i+1})$ and depth $d$. Let $y_1, \ldots, y_{32n}$ be the gates on layer $d+1$, defined by
\[
y_j = \sum_{i=1}^t \alpha_{j,i} \, C_i(x)_j,
\]
where $j = 1, \ldots, 32n$ and the coefficients $\alpha_{j,i} \in \mathbb{F}_q$ are chosen uniformly at random from $\mathbb{F}_q$. A union bound argument shows that there exist coefficients $\alpha_{j,i}$ such that $\mathrm{wt}(y) \ge \frac{n}{4}$ for all $x \in \mathbb{F}_q^n$ with $\weight(x) \in [r_1, r_{t+1}]$.

By applying the rate amplifier, the distance can be increased from $n/4$ to $4n$, producing a circuit of depth $d+2$. Since the rate amplifier's output gates have bounded fan-in $O(1)$, the final layer can be collapsed, resulting in a circuit of depth $d+1$ and size $\sum_{i=1}^t S_d(n, r_i, r_{i+1}) + O(t n).$

Assume that each $(n, r_i, r_{i+1})$-PGC has output gates of bounded fan-in $D$. Collapsing the last layer, we obtain a circuit of depth $d$ and size $\sum_{i=1}^t S_d(n, r_i, r_{i+1}) + O(Dtn)$.
\end{proof}

The following lemma introduces an \emph{output amplifier}, which increases 
the number of output coordinates while preserving a relative distance of 
at least $1/8$.

\begin{lemma} [Output Amplifier] 
\label{lem:output_amplifier}
Fix a finite field $\mathbb{F}_q$. For every positive integer $n$ and every $m \ge 3n$, there exists an 
$(n, m, n/8, n, m/8)$-range detector over $\mathbb{F}_q$ that can be computed by a depth-1 
linear circuit of size $O(m)$. In addition, each output gate has fan-in bounded by an absolute constant.
\end{lemma}
\begin{proof}
Identical to Lemma~3.8 of~\cite{drucker2023minimum}, with the $\mathbb{F}_2$ 
union bound replaced by inequality~\eqref{equ:bsum_fq}.
\end{proof}

%The proof of Lemma \ref{lem:output_amplifier} follows the same argument as Lemma 3.8 of \cite{drucker2023minimum}, using a union bound together with inequality \eqref{equ:bsum_fq}.

\begin{lemma}[Condenser~{\cite{GHK+12}}]
\label{lem:condenser}
Fix a field $\mathbb{F}_q$. There exists a constant $c_0 = c_0(q)$ such 
that for all integers $n$, real numbers $r, s$ with $c_0 \le r \le n$ 
and $s \in [1, n/r^{1.5}]$, there exists an
\[
(n,\, \lfloor n/r \rfloor,\, s,\, n/r^{1.5},\, s,\, \lfloor n/r \rfloor)
\text{-range detector}
\]
computable by a depth-$1$ linear circuit of size $O(n)$.
\end{lemma}
\begin{proof}
The proof closely follows Lemma~23 of~\cite{GHK+12}, with the only 
difference being the verification that for $r \ge c_0(q)$ sufficiently 
large,
\[
\left(\frac{6e\ell}{\frac{5}{6}\cdot m}\right)^{\frac{5}{6}\cdot 6\ell} 
\le 2^{-\ell}(q-1)^{-\ell}\cdot \binom{n}{\ell},
\]
where $m = \lfloor n/r \rfloor$ and $\ell \in [r,\, n/r^{1.5}]$.
\end{proof}

\begin{lemma} [Reduction Lemma]
\label{lem:reduction}
Fix a finite field $\mathbb{F}_q$. Let $c_0 = c_0(q)$ be the constant in Lemma \ref{lem:condenser}.
 For any $r \in [c_0, n]$ and $1 \le s \le t \le \frac{n}{r^{1.5}}$,
\begin{equation}
\label{equ:red_ineq}
S_d(n, s, t) \,\,\le\,\, S_{d-2}\left(\left\lfloor \frac{n}{r} \right\rfloor, s, \frac{n}{r}\right) + O(n).
\end{equation}
In addition, the output gates computing the $(n, s, t)$-PGC have bounded fan-in $O(1)$.
\end{lemma}

The proof of Lemma \ref{lem:reduction} is the same as that of Lemma 3.9 in \cite{drucker2023minimum}.

\begin{lemma} [Depth-$2$ Base Case]
\label{lem:depth2}
Fix $\mathbb{F}_q$. For any $r \in [1, n]$, we have
\[
S_2\left(n, \frac{n}{r}, n\right) = O_q\left(\log^2 r \cdot n\right).
\]
\end{lemma}
\begin{proof}
The proof of Lemma \ref{lem:depth2} follows the same argument as in the proof of Lemma 27 in \cite{GHK+12}; we provide a brief sketch here.

Let $k_1 = r$ and define $k_{i+1} = k_i / 2$. Let $t$ be the smallest integer such that $k_t \le 1$.  
Our strategy is to first construct $(n, n/k_i, n/k_{i+1})$-PGCs of size $O_q(n \log r)$, and then use the Composition Lemma to combine $O(\log r)$ such PGCs.

We construct an $(n, n/k_i, n/k_{i+1})$-PGC as follows. Let $n_i = O\bigl(\frac{n}{k_i} \cdot \log r\bigr)$, and let the middle layer be $y_1, \ldots, y_{n_i}$. Each $y_j$ is connected to $O(k_i)$ inputs chosen independently and uniformly at random with replacement, with coefficients drawn uniformly from $\mathbb{F}_q$. One can argue that $\Pr[y_j = 0] \le 1/2$. Applying a Chernoff bound, we obtain
\[
\Pr\bigl[\weight(y) \le \frac{n_i}{8}\bigr] \le 2^{-\frac{3}{64} \cdot n_i} <
\sum_{j \le n/k_{i+1}} \binom{n}{j}.
\]
A union bound then implies the existence of a linear circuit such that $\weight(y) \ge n_i/8$ for all $y \in \mathbb{F}_q^n$ with $\weight(y) \in [n/k_i, n/k_{i+1}]$. Finally, by placing an output amplifier (Lemma \ref{lem:output_amplifier}) at the bottom, we obtain an $(n, n/k_i, n/k_{i+1})$-PGC of depth $2$, and by Lemma \ref{lem:output_amplifier}, the fan-in of each output gate is bounded by $O(1)$.

Applying the Composition Lemma~\ref{lem:composition} to compose the $t$ 
PGCs, $(n, n/k_i, n/k_{i+1})$ for $i = 1, \ldots, t$, gives an $(n, n/r, n)$-PGC of depth $2$ and size 
$O_q(n \log^2 r)$.
\end{proof}

\begin{lemma} [Depth 4 Base Case]
\label{lem:depth4}
Fix a finite field $\mathbb{F}_q$. For any $r \in [1, n]$, we have
\[
S_4\left(n, \frac{n}{r}, n\right) = O_q\left(\lambda_2(n)\cdot n\right).
\]
\end{lemma}
\begin{proof}
The proof follows Lemma~26 of~\cite{GHK+12}; we sketch the key steps.

Let $c_0 = c_0(q)$ be the constant from Lemma \ref{lem:condenser}.  
If $r < c_0$, then $S_4(n, n/r, n) = O_q(n)$ since a distance amplifier 
(Lemma~\ref{lem:rate_booster}) suffices.
Hence, we may assume $r \ge c_0$.  

Let $k_1 = c_0$ and define $k_{i+1} = 2^{\sqrt{k_i}}$, and let $t$ be the smallest integer such that $k_t \ge n$.  
Note that $k_{i+2} \ge 2^{k_i}$, which implies that $t = O(\log^* n) = O(\lambda_2(n))$.

By Lemma \ref{lem:depth2}, we have
\begin{align*}
S_2\Bigl(\frac{n}{k_i^{2/3}}, \frac{n}{k_{i+1}}, \frac{n}{k_i}\Bigr) 
&= S_2\Bigl(\frac{n}{k_i^{2/3}}, \frac{n}{2^{\sqrt{k_i}}}, \frac{n}{k_i}\Bigr) \\
&\le \frac{n}{k_i^{2/3}} \, O_q\Bigl(\log^2 2^{\sqrt{k_i}}\Bigr) \\
&= O_q(n).
\end{align*}
Applying Lemma \ref{lem:reduction}, we then obtain an $(n, n/k_{i+1}, n/k_i)$-PGC computable by a depth-$4$ linear circuit with output fan-in bounded by $O(1)$.

Finally, by applying the Composition Lemma (Lemma \ref{lem:composition}) to combine the $O(\log^* n)$ PGCs, we obtain a $(n, n/r, n)$-PGC of depth $4$ and size $O_q(n \log^* n) = O_q(\lambda_2(n) \cdot n)$.
\end{proof}

The following theorem provides the main construction and is proved by induction on $k$. It is established for $q = 2$ in Theorem~3.13 of~\cite{drucker2023minimum}. The argument extends to a general finite field $\mathbb{F}_q$; we omit the details.

\begin{theorem}
\label{thm:ub_main}
Fix a finite field $\mathbb{F}_q$. Let $c_0 = c_0(q)$ be the constant 
from Lemma~\ref{lem:condenser}. There exist constants $c, D > 0$, 
depending only on $q$, such that the following statements hold.
\begin{enumerate}
    \item For any $c_0 \le r \le n$ and any $k \ge 3$,
    \begin{equation}
    \label{equ:d2k_first}
    S_{2k}\!\left(n,\, \frac{n}{A(k-1,r)},\, \frac{n}{r}\right) \le 2cn,
    \end{equation}
    and the output gates of the corresponding linear circuit have fan-in 
    bounded by $D$.
    \item For any $2 \le r \le n$ and any $k \ge 2$,
    \begin{equation}
    \label{equ:d2k_second}
    S_{2k}\!\left(n,\, \frac{n}{r},\, n\right) \le 3c\,\lambda_k(r) \cdot n.
    \end{equation}
    The linear circuits encoding the $(n, n/r, n)$-PGC do not necessarily 
    have bounded output fan-in.
\end{enumerate}
\end{theorem}

Theorem \ref{thm:ub_main} immediately implies Theorem \ref{thm:mother_code}.

%% file: additive_channel_proof.tex
\section{Existence of Capacity-Achieving Codes}

\subsection{Code Construction}
\label{section:code_construction}

Let $H = (L=[32k], R = [n], E)$ be a $(1/8, \gamma)$-disperser graph for some small constant $\gamma > 0$. Recall that the disperser property guarantees $|N(S)| \ge (1-\gamma)n$ for every subset $S \subseteq L$ with $|S| \ge |L|/8$. Given an edge-labeling $\alpha: E(H) \to \mathbb{F}_q$, define the linear map $D_{H,\alpha}: \mathbb{F}_q^{32k} \to \mathbb{F}_q^n$ by
\begin{equation}
    D_{H, \alpha}(x_1, \dots, x_{32k})_j = \sum_{(i,j) \in E(H)} \alpha(i,j)\, x_i,
\end{equation}
where all arithmetic is over $\mathbb{F}_q$.

Let $\basecode:\mathbb{F}_q^k \to \mathbb{F}_q^{32k}$ be a linear code with minimum distance at least $4k$. 
The encoder $\enc:\mathbb{F}_q^k \to \mathbb{F}_q^n$ is defined by
\begin{equation}
    \enc(x) = D_{H, \alpha}\bigl(\basecode(x)\bigr).
\end{equation}

Let $\typical_\epsilon(y)$ denote the set of typical transmitted vectors corresponding to the received word $y \in \mathbb{F}_q^n$ for a small constant $\epsilon > 0$  (see Definition~\ref{def:typical_set} below). The decoder outputs the unique codeword in $\typical_\epsilon(y)$ if it exists, and \textsf{fail} otherwise. Since our goal is an existential result, decoding efficiency is not a concern.

\subsection{Typical Set}

For a received vector $y$, its conditional typical set consists of those vectors $x$ for which the channel behavior --- transmitting $x$ and receiving $y$ --- is \emph{statistically typical}, i.e., consistent with what one would normally expect from the channel.

\begin{definition}[Conditional Typical Set]
\label{def:typical_set}
Let $\Pi$ be a strongly symmetric channel. For a received sequence $y \in \mathbb{F}_q^n$, define the \emph{conditional typical set} as
\[
\mathrm{Typical}_\epsilon(y) = \left\{ x \in \mathbb{F}_q^n :
p(y \mid x) > 0
\quad \text{and} \quad
\left| -\frac{1}{n} \sum_{i=1}^n \log p(y_i \mid x_i) - H_\Pi \right| \leq \epsilon
\right\},
\]
where $H_\Pi$ is the per-row output entropy of $\Pi$.
\end{definition}

\begin{proposition}
\label{prop:typical_volume}
For any $y \in \mathbb{F}_q^n$ and any $\epsilon > 0$,
\[
|\mathrm{Typical}_\epsilon(y)| \leq 2^{n(H_\Pi + \epsilon)}.
\]
\end{proposition}

\begin{proof}
By definition of $\mathrm{Typical}_\epsilon(y)$, every $x \in \mathrm{Typical}_\epsilon(y)$ satisfies
\[
p(y \mid x) \geq 2^{-n(H_\Pi + \epsilon)}.
\]
Therefore,
\[
1 = \sum_{x \in \mathbb{F}_q^n} p(y \mid x) \geq \sum_{x \in \mathrm{Typical}_\epsilon(y)} p(y \mid x) \geq |\mathrm{Typical}_\epsilon(y)| \cdot 2^{-n(H_\Pi + \epsilon)},
\]
which gives $|\mathrm{Typical}_\epsilon(y)| \leq 2^{n(H_\Pi + \epsilon)}$.
\end{proof}

\begin{lemma}[Chernoff bound]
Let $X_1, X_2, \ldots, X_n \in \{0,1\}$ be independent random variables with
$\mathbb{E}[X_i] = p$, and let $S_n = \sum_{i=1}^n X_i$. 
Then for any $\epsilon > 0$,
\[
\Pr\bigl[\, |S_n - pn| \ge \epsilon n \,\bigr]
\le 2 \exp\!\left( - \frac{\epsilon^2 pn}{2} \right).
\]
\end{lemma}

The following lemma is a quantitative version of the Asymptotic Equipartition Property \cite{CoverThomas2006}; we include the proof for completeness as the exponential bound will be needed later.

\begin{lemma}
\label{lem:typical_prob}
For any $x \in \mathbb{F}_q^n$ and $\epsilon > 0$,
\[
\Pr_{y \sim p(\cdot \mid x)}\!\bigl[x \notin \typical_\epsilon(y)\bigr] \;\le\; 2q\cdot 2^{-\Omega_\Pi(\epsilon^2 n)}.
\]
\end{lemma}
\begin{proof}
Since $\Pi$ is strongly symmetric, the distribution of $Z_i = -\log p(y_i \mid x_i)$
is the same for any fixed $x_i$, and the entropy is the same for every
row. Therefore, without loss of generality, assume $x = 0^n$, so that
$y_1, \ldots, y_n$ are i.i.d.\ with common distribution $p(\cdot \mid 0)$.

Write $p_a = p(a \mid 0)$ for $a \in \mathbb{F}_q$, and for each $a$ let
$X_i^{(a)} = \mathbf{1}[y_i = a]$, so that
\[
\frac{1}{n}\sum_{i=1}^n Z_i
= \sum_{a \in \mathbb{F}_q} (-\log p_a)\cdot \frac{1}{n}\sum_{i=1}^n X_i^{(a)}.
\]
Taking expectations gives $\mathbb{E}[Z_i] = H_\Pi$, so
\[
\frac{1}{n}\sum_{i=1}^n Z_i - H_\Pi
= \sum_{a \in \mathbb{F}_q}(-\log p_a)
  \left(\frac{1}{n}\sum_{i=1}^n X_i^{(a)} - p_a\right).
\]
If the left-hand side exceeds $\epsilon$ in absolute value, then by the triangle
inequality there exists some $b \in \mathbb{F}_q$ with
\[
\left|\frac{1}{n}\sum_{i=1}^n X_i^{(b)} - p_{b}\right|
> \frac{\epsilon}{q \cdot \max_b |\log p_b|} =: \epsilon' > 0.
\]
Since $\{X_i^{(a)}\}_i$ are i.i.d.\ Bernoulli$(p_a)$, the Chernoff bound and a
union bound over all $q$ symbols give
\[
\Pr\!\left[\,\left|\frac{1}{n}\sum_{i=1}^n Z_i - H_\Pi\right| > \epsilon\,\right]
\leq \sum_{a \in \mathbb{F}_q} 2\exp\!\left(-\frac{(\epsilon')^2 p_a\, n}{2}\right)
\leq 2q \cdot \exp(-c_\Pi\,\epsilon^2 n),
\]
where 
\[
c_\Pi = \frac{\min_a p_a}{2(q \cdot \max_a |\log p_a|)^2} > 0
\]
depends only on $\Pi$. Therefore,
\[
\Pr_{y \sim p(\cdot \mid x)}\!\bigl[x \notin \mathrm{Typical}_\epsilon(y)\bigr]
\leq 2q \cdot 2^{-\Omega_\Pi(\epsilon^2 n)}.
\qedhere
\]
\end{proof}

\subsection{Bounding the Error Probability}

In this section, we prove Theorem~\ref{thm:main}. The encoder composes two layers: a \emph{mother code} $\basecode: \mathbb{F}_q^k \to \mathbb{F}_q^{32k}$ with minimum distance $4k$, and a \emph{disperser layer} $D: \mathbb{F}_q^{32k} \to \mathbb{F}_q^n$. The disperser layer is a bipartite disperser graph with edge weights drawn uniformly and independently at random from $\mathbb{F}_q$.
For any nonzero $m \in \mathbb{F}_q^k$, the mother code guarantees $|\supp(\basecode(m))| \ge 4k$; the disperser layer then ensures that $D(\basecode(m))$ is uniformly distributed over a subspace of $\mathbb{F}_q^n$ of dimension at least $(1-\gamma)n$, for a small constant $\gamma > 0$. This mimics the behavior of a random linear code, which is known to achieve channel capacity.

The additive noise structure of $\Pi$ is used in two places. First, it 
ensures that the typical set has the correct volume and probability 
(Proposition~\ref{prop:typical_volume} and Lemma~\ref{lem:typical_prob}). 
Second, it implies that $P_e(m) = P_e(0)$ for every message $m \in 
\mathbb{F}_q^k$ (Claim~\ref{claim:symmetry}), so that it suffices to 
bound the error probability for a single message.

\begin{proof}[Proof of Theorem \ref{thm:main}]
Let $\basecode:\mathbb{F}_q^k \to \mathbb{F}_q^{32k}$ be a fixed linear code 
with minimum distance at least $4k$ (Theorem~\ref{thm:mother_code}). Let $H = (L=[32k],\, R=[n],\, E)$ be a fixed $(1/8,\gamma)$-disperser graph, where $\gamma > 0$ is a small constant to be determined later.

Let $\alpha:E(H) \to \mathbb{F}_q$ assign a field element to each edge. Define the linear map $D_{H,\alpha}: \mathbb{F}_q^{32k} \to 
\mathbb{F}_q^n$ by
\[
    D_{H,\alpha}(x)_j \;=\; \sum_{(i,j)\,\in\, E(H)} \alpha(i,j)\, x_i.
\]

The encoder $\enc:\mathbb{F}_q^k \to \mathbb{F}_q^n$ is defined as the 
two-step composition
\[
    \enc(x) \;=\; D_{H,\alpha}\!\bigl(\basecode(x)\bigr),
\]
where the message is first expanded to $\mathbb{F}_q^{32k}$ via $\basecode$, then mapped 
to $\mathbb{F}_q^n$ via $D_{H,\alpha}$.

\textbf{Proof strategy.} Rather than constructing a good code explicitly, we establish its 
\emph{existence} via a probabilistic argument. Specifically, we fix the  bipartite graph $H$ and randomize only the edge labeling: let $\dot\alpha: E(H) \to \mathbb{F}_q$ be chosen 
uniformly at random, inducing a random encoder
\[
    \dot\enc(x) \;=\; D_{H,\dot\alpha}\!\bigl(\basecode(x)\bigr).
\]
This is equivalently described by a random generator matrix $\dot G = D_{H,\dot\alpha}\,\basecode$, so that $\dot\enc(x) = \dot G x.$

It then suffices to show that the \emph{expected} decoding error under this 
random choice of $\dot G$ is small. By a standard averaging argument, this implies the existence of at least one fixed choice of $\alpha$ --- and hence a deterministic encoder $\enc$ --- that achieves the desired error guarantee.

Let $\dot G \in \mathbb{F}_q^{n \times k}$ be the random generator matrix. Let $m \in \mathbb{F}_q^k$ be the message, and $c = \dot G m \in \mathbb{F}_q^n$ be the codeword.

\begin{claim}
\label{claim:symmetry}
Let $\Pi$ be an additive noise channel over $\mathbb{F}_q$ and let 
$G \in \mathbb{F}_q^{n \times k}$ be any fixed generator matrix. Then 
$P_e(m) = P_e(0)$ for every $m \in \mathbb{F}_q^k$.
\end{claim}
\begin{proof} (of Claim \ref{claim:symmetry})
Define the bijection $\sigma : \mathbb{F}_q^n \to \mathbb{F}_q^n$ by 
$\sigma(z) = z - Gm$, and write $z' = 
\sigma(z)$. Since $\sigma$ is a bijection, substituting 
$z \mapsto z'$ in a sum over $\mathbb{F}_q^n$ is valid. 
We verify three facts.

\medskip
\noindent\textbf{Fact 1.} $p(z \mid Gm) = p(z' \mid 0)$.

Since $\Pi$ is an additive noise channel, $p(z_i \mid x_i) = p_N(z_i - x_i)$, 
hence by independence across coordinates,
\[
p(z \mid Gm) 
= \prod_{i=1}^n p_N(z_i - (Gm)_i) 
= \prod_{i=1}^n p_N(z_i') 
= p(z' \mid 0).
\]

\medskip
\noindent\textbf{Fact 2.} $\mathrm{Typical}_\epsilon(z) - Gm = 
\mathrm{Typical}_\epsilon(z')$.

For any $x \in \mathbb{F}_q^n$, by the additive noise structure,
\[
p(z_i \mid x_i) = p_N(z_i - x_i) = p_N(z_i' - (x_i - (Gm)_i)) 
= p(z_i' \mid x_i - (Gm)_i),
\]
hence
\[
-\frac{1}{n}\sum_{i=1}^n \log p(z_i \mid x_i) 
= -\frac{1}{n}\sum_{i=1}^n \log p(z_i' \mid x_i - (Gm)_i).
\]
Therefore $x \in \mathrm{Typical}_\epsilon(z) \iff 
x - Gm \in \mathrm{Typical}_\epsilon(z')$, which gives 
$\mathrm{Typical}_\epsilon(z) - Gm = \mathrm{Typical}_\epsilon(z')$. In particular $Gm \in \mathrm{Typical}_\epsilon(z) \iff 
0 \in \mathrm{Typical}_\epsilon(z')$.

\medskip
\noindent\textbf{Fact 3.} $\exists\, m' \neq m : Gm' \in 
\mathrm{Typical}_\epsilon(z) \iff \exists\, \hat{m} \neq 0 : 
G\hat{m} \in \mathrm{Typical}_\epsilon(z')$.

By Fact 2, $Gm' \in \mathrm{Typical}_\epsilon(z) \iff Gm' - Gm 
\in \mathrm{Typical}_\epsilon(z')$. Setting $\hat{m} = m' - m$ 
and using linearity of $G$, we have $Gm' - Gm = G\hat{m}$, and 
$m' \neq m \iff \hat{m} \neq 0$.

\medskip
Therefore, applying the substitution $z \mapsto 
z' = z - Gm$ and the three facts above, we have
\begin{align*}
P_e(m) 
&= \sum_{z} p(z \mid Gm)\cdot\mathbf{1}\!\Big[
    Gm \notin \mathrm{Typical}_\epsilon(z)\ \vee\ 
    \exists\, m' \neq m : Gm' \in \mathrm{Typical}_\epsilon(z)
\Big] \\
&= \sum_{z'} p(z' \mid 0)\cdot\mathbf{1}\!\Big[
    0 \notin \mathrm{Typical}_\epsilon(z')\ \vee\ 
    \exists\, \hat{m} \neq 0 : G\hat{m} \in \mathrm{Typical}_\epsilon(z')
\Big] \\
&= P_e(0). \qedhere
\end{align*}
\end{proof}

Claim~\ref{claim:symmetry} reduces the analysis to bounding $P_e(0)$; it thus suffices to exhibit $G$ for which $P_e(0)$ is small.

\noindent\textbf{Step 1: Expressing the error probability.}  
By the decoder defined in Section~\ref{section:code_construction}, the error event for message $0 \in \mathbb{F}_q^k$ consists of either a typicality failure or a collision with another codeword:
\[
P_e(0) = \sum_{z \in \mathbb{F}_q^n} p(z \mid 0) \, 
\mathbf{1}\Big[ 0 \notin \typical_\epsilon(z) \ \vee \ \exists m \neq 0 : \dot G m \in \typical_\epsilon(z) \Big].
\]
By linearity of expectation over a random generator matrix $\dot G = \basecode C_{H, \dot\alpha}$ (with $\alpha : E(H) \to \mathbb{F}_q$ uniform), we have
\[
\mathbb{E}_{\dot G}[P_e(0)]
\le \sum_{z} p(z \mid 0) \mathbf{1}[0 \notin \typical_\epsilon(z)] 
+ \sum_{m \neq 0} \sum_{z} p(z \mid 0) \Pr[\dot G m \in \typical_\epsilon(z)].
\]
Define
\[
E_1 = \sum_{z} p(z \mid 0) \mathbf{1}[0 \notin \typical_\epsilon(z)], \qquad
E_2 = \sum_{m \neq 0} \sum_{z} p(z \mid 0) \Pr[\dot G m \in \typical_\epsilon(z)].
\]

\noindent\textbf{Step 2: Bounding $E_1$.}  
By Lemma \ref{lem:typical_prob}, we have $E_1 \le 2q \cdot 2^{-\Omega_\Pi(\epsilon^2 n)}$.

\noindent\textbf{Step 3: Bounding $E_2$.}  
We bound the probability that some nonzero $m \in \mathbb{F}_q^k$ produces a codeword in $\typical_\epsilon(z)$.

Since $\basecode$ has minimum distance at least $4k$, every nonzero $m$ satisfies
$|\supp(\basecode(m))| \ge 4k = |L|/8$.
By the
$(1/8,\gamma)$-disperser property of $H$, there exists $S \subseteq [n]$ with
$|S| \ge (1-\gamma)n$ such that every $v \in S$ has at least one neighbor
$u \in \supp(\basecode(m))$. For each such $v$, fix any such edge $e=(u,v)$:
since $\basecode(m)_u \neq 0$ and $\dot\alpha(e)$ is uniform over $\mathbb{F}_q$,
the output $(\dot G m)_v$ is uniform over $\mathbb{F}_q$. Independence of edge
weights across edges implies these outputs are independent across $v \in S$, so
\[
    (\dot G m)\restriction_S \;\sim\; \mathrm{Uniform}(\mathbb{F}_q^S).
\]

By Proposition~\ref{prop:typical_volume}, $|\typical_\epsilon(z)| \le 2^{n(H_\Pi+\epsilon)}$, hence $|\typical_\epsilon(z)\restriction_S| \le 2^{n(H_\Pi+\epsilon)}$. Thus,
\begin{align*}
    \Pr\bigl[\dot G m \in \typical_\epsilon(z)\bigr]
    &\le \Pr\bigl[(\dot G m)\restriction_S \in \typical_\epsilon(z)\restriction_S\bigr] \\
    &\le \frac{2^{n(H_\Pi+\epsilon)}}{q^{(1-\gamma)n}}.
\end{align*}
Summing over all $q^k - 1 \le q^k$ nonzero $m$ gives
\[
    E_2 \;\le\; \frac{q^k \cdot 2^{n(H_\Pi+\epsilon)}}{q^{(1-\gamma)n}}
    \;=\; 2^{\,n\left(r + H_\Pi + \epsilon - (1-\gamma)\log q\right)},
\]
where $r = \frac{k}{n}\log q$ is the code rate. The exponent equals
\begin{equation}
\label{equ:error_exponent}
    n\bigl(r - \log q + H_\Pi + \gamma\log q + \epsilon\bigr).
\end{equation}
For any fixed rate $r < \log q - H_\Pi$ (i.e., below the capacity
\eqref{equ:sym_channel_capacity}), choosing $\epsilon, \gamma > 0$ sufficiently small
makes the exponent $-\Omega_{\Pi,r}(n)$.

\noindent\textbf{Step 4: Conclusion.}  
Combining the bounds for $E_1$ and $E_2$, we conclude that
\[
\mathbb{E}_{\dot G}[P_e(0)] \le E_1 + E_2 \le 2^{- \Omega_{\Pi, r}(n)}.
\]
By averaging, there exists $G$ such that $P_e(0) \le 2^{-\Omega_{\Pi,r}(n)}$. 
Claim~\ref{claim:symmetry} then gives $P_e(m) = P_e(0)$ for all $m \in \mathbb{F}_q^k$. Thus, the worst-case error probability is bounded by $2^{-\Omega_{\Pi,r}(n)}$.

\textbf{Circuit size and depth.} By Theorem~\ref{thm:mother_code}, the mother code $\basecode$ is encodable by a linear circuit of depth $2\alpha(k)$ and size $O_q(k)$, where $k = rn$.
The disperser code $D_{H,\alpha}$ is computed by a linear circuit of depth $1$ and size $O(n)$ with bounded output degree $O_\gamma(1) = O_{\Pi,r}(1)$ (Theorem~\ref{thm:prob_disperser}). Composing these and collapsing the final layer yields a linear circuit of depth $2\alpha(n)$ and size $O_{\Pi,r}(n)$.

\end{proof}

%% file: conclusion.tex
\section{Conclusion}

We have shown that for additive noise channels over $\mathbb{F}_q$ with $q$ a prime power, there exist capacity-achieving codes encodable by linear circuits over $\mathbb{F}_q$ of linear size and inverse-Ackermann depth. 
While the noisy coding theorem demonstrates the abundance of capacity-achieving codes, our result shows that some of them admit encoders of extremely low circuit complexity.

This leaves a natural open problem: to construct error-correcting codes that achieve capacity, admit linear-size circuits of inverse-Ackermann depth for encoding, and also support efficient decoding.

\subsection*{Acknowledgements}

This work was supported by the CCF-Huawei Populus Grove Fund.

%% file: code.bib
@article{reeves2023reed,
  title={Reed--Muller codes on BMS channels achieve vanishing bit-error probability for all rates below capacity},
  author={Reeves, Galen and Pfister, Henry D},
  journal={IEEE Transactions on Information Theory},
  volume={70},
  number={2},
  pages={920--949},
  year={2023},
  publisher={IEEE}
}

@inproceedings{AS23,
  author    = {Abbe, Emmanuel and Sandon, Colin},
  title     = {A Proof that Reed--Muller Codes Achieve Shannon Capacity on Symmetric Channels},
  booktitle = {Proc. IEEE 64th Annual Symposium on Foundations of Computer Science (FOCS)},
  pages     = {177--193},
  year      = {2023}
}

@article{KLP12,
  author  = {Kaufman, Tali and Lovett, Shachar and Porat, Ely},
  title   = {Weight Distribution and List-Decoding Size of Reed--Muller Codes},
  journal = {IEEE Transactions on Information Theory},
  volume  = {58},
  number  = {5},
  pages   = {2689--2696},
  year    = {2012}
}

@inproceedings{ASW15,
  author    = {Abbe, Emmanuel and Shpilka, Amir and Wigderson, Avi},
  title     = {Reed--Muller Codes for Random Erasures and Errors},
  booktitle = {Proc. 47th Annual ACM Symposium on Theory of Computing (STOC)},
  pages     = {297--306},
  year      = {2015}
}

@inproceedings{KKM+16,
  author    = {Kudekar, Shrinivas and Kumar, Santhosh and Mondelli, Marco and Pfister, Henry D. and Urbanke, R{\"u}diger},
  title     = {Comparing the Bit-MAP and Block-MAP Decoding Thresholds of Reed--Muller Codes on BMS Channels},
  booktitle = {Proc. IEEE International Symposium on Information Theory (ISIT)},
  pages     = {1755--1759},
  year      = {2016}
}

@book{CoverThomas2006,
  author    = {Thomas M. Cover and Joy A. Thomas},
  title     = {Elements of Information Theory},
  edition   = {2nd},
  publisher = {Wiley-Interscience},
  year      = {2006},
}

@article{Arikan09_polar,
  title={Channel polarization: A method for constructing capacity-achieving codes for symmetric binary-input memoryless channels},
  author={Arikan, Erdal},
  journal={IEEE Transactions on information Theory},
  volume={55},
  number={7},
  pages={3051--3073},
  year={2009},
  publisher={IEEE}
}

@article{li2003linear,
  title={Linear network coding},
  author={Li, S-YR and Yeung, Raymond W and Cai, Ning},
  journal={IEEE transactions on information theory},
  volume={49},
  number={2},
  pages={371--381},
  year={2003},
  publisher={IEEE}
}

@article{feinstein1954new,
  title={A new basic theorem of information theory},
  author={Feinstein, Amiel},
  year={1954},
  publisher={Research Laboratory of Electronics, Massachusetts Institute of Technology}
}

@book{mackay2003information,
  title={Information theory, inference and learning algorithms},
  author={MacKay, David JC},
  year={2003},
  publisher={Cambridge university press}
}

@book{gallager1968information,
  title={Information theory and reliable communication},
  author={Gallager, Robert G},
  volume={588},
  year={1968},
  publisher={Springer}
}

@inproceedings{guruswami2001expander,
  title={Expander-based constructions of efficiently decodable codes},
  author={Guruswami, Venkatesan and Indyk, Piotr},
  booktitle={Proceedings 42nd IEEE Symposium on Foundations of Computer Science},
  pages={658--667},
  year={2001},
  organization={IEEE}
}

@article{rom2006improving,
  title={Improving the alphabet-size in expander-based code constructions},
  author={Rom, Eran and Ta-Shma, Amnon},
  journal={IEEE transactions on information theory},
  volume={52},
  number={8},
  pages={3695--3700},
  year={2006},
  publisher={IEEE}
}

@inproceedings{druk2014linear,
  title={Linear-time encodable codes meeting the gilbert-varshamov bound and their cryptographic applications},
  author={Druk, Erez and Ishai, Yuval},
  booktitle={Proceedings of the 5th conference on Innovations in theoretical computer science},
  pages={169--182},
  year={2014}
}

@inproceedings{guruswami2004better,
  title={Better extractors for better codes?},
  author={Guruswami, Venkatesan},
  booktitle={Proceedings of the thirty-sixth annual ACM symposium on Theory of computing},
  pages={436--444},
  year={2004}
}

@inproceedings{li2025improved,
  title={Improved Explicit Near-Optimal Codes in the High-Noise Regimes},
  author={Li, Xin and Mao, Songtao},
  booktitle={Proceedings of the 2025 Annual ACM-SIAM Symposium on Discrete Algorithms (SODA)},
  pages={5560--5581},
  year={2025},
  organization={SIAM}
}

@inproceedings{berrou1993near,
  title={Near Shannon limit error-correcting coding and decoding: Turbo-codes. 1},
  author={Berrou, Claude and Glavieux, Alain and Thitimajshima, Punya},
  booktitle={Proceedings of ICC'93-IEEE International Conference on Communications},
  volume={2},
  pages={1064--1070},
  year={1993},
  organization={IEEE}
}

@article{LM09_linear_ldpc,
  title={Linear time encoding of LDPC codes},
  author={Lu, Jin and Moura, Jos{\'e} MF},
  journal={IEEE Transactions on Information Theory},
  volume={56},
  number={1},
  pages={233--249},
  year={2009},
  publisher={IEEE}
}

@article{RU01_LDPC_encode,
  title={Efficient encoding of low-density parity-check codes},
  author={Richardson, Thomas J and Urbanke, R{\"u}diger L},
  journal={IEEE transactions on information theory},
  volume={47},
  number={2},
  pages={638--656},
  year={2001},
  publisher={IEEE}
}

@inproceedings{Luby1997_practicalldpc,
  title={Practical loss-resilient codes},
  author={Luby, Michael G and Mitzenmacher, Michael and Shokrollahi, M Amin and Spielman, Daniel A and Stemann, Volker},
  booktitle={Proceedings of the twenty-ninth annual ACM symposium on Theory of computing},
  pages={150--159},
  year={1997}
}

@article{stockmeyer1984simulation,
  title={Simulation of parallel random access machines by circuits},
  author={Stockmeyer, Larry and Vishkin, Uzi},
  journal={SIAM Journal on Computing},
  volume={13},
  number={2},
  pages={409--422},
  year={1984},
  publisher={SIAM}
}

@article{shannon1948mathematical,
  title={A mathematical theory of communication},
  author={Shannon, Claude E},
  journal={The Bell system technical journal},
  volume={27},
  number={3},
  pages={379--423},
  year={1948},
  publisher={Nokia Bell Labs}
}

@article{guruswami2012essential,
  title={Essential coding theory},
  author={Guruswami, Venkatesan and Rudra, Atri and Sudan, Madhu},
  journal={Draft available at http://www. cse. buffalo. edu/atri/courses/coding-theory/book},
  volume={2},
  number={1},
  year={2012}
}

@inproceedings{drucker2023minimum,
  title={On the Minimum Depth of Circuits with Linear Number of Wires Encoding Good Codes},
  author={Drucker, Andrew and Li, Yuan},
  booktitle={International Computing and Combinatorics Conference},
  pages={392--403},
  year={2023},
  organization={Springer}
}

@inproceedings{DDP+83,
  author    = {Danny Dolev and
               Cynthia Dwork and
               Nicholas Pippenger and
               Avi Wigderson},
  title     = {Superconcentrators, generalizers and generalized connectors with limited
               depth},
  booktitle = {Proceedings of the 15th Annual {ACM} Symposium on Theory of Computing,
               25-27 April, 1983, Boston, Massachusetts, {USA}},
  pages     = {42--51},
  year      = {1983},
  crossref  = {DBLP:conf/stoc/STOC15},
%  url       = {http://doi.acm.org/10.1145/800061.808731},
%  doi       = {10.1145/800061.808731},
  timestamp = {Mon, 17 Oct 2011 17:25:06 +0200},
  biburl    = {http://dblp.org/rec/bib/conf/stoc/DolevDPW83},
  bibsource = {dblp computer science bibliography, http://dblp.org}
}

@article{RT00,
  title     = {Bounds for dispersers, extractors, and depth-two superconcentrators},
  author    = {Radhakrishnan, Jaikumar and Ta-Shma, Amnon},
  journal   = {SIAM Journal on Discrete Mathematics},
  volume    = {13},
  number    = {1},
  pages     = {2--24},
  year      = {2000},
  publisher = {SIAM}
}

@article{RS03,
  title     = {Lower bounds for matrix product in bounded depth circuits with arbitrary gates},
  author    = {Raz, Ran and Shpilka, Amir},
  journal   = {SIAM Journal on Computing},
  volume    = {32},
  number    = {2},
  pages     = {488--513},
  year      = {2003},
  publisher = {SIAM}
}

@article{GHK+12,
  title     = {Tight Bounds on Computing Error-Correcting Codes by Bounded-Depth Circuits With Arbitrary Gates},
  author    = {Gal, Anna and Hansen, Kristoffer Arnsfelt and Koucky, Michal and Pudlak, Pavel and Viola, Emanuele},
  journal   = {IEEE Transactions on Information Theory},
  volume    = {59},
  number    = {10},
  pages     = {6611--6627},
  year      = {2013},
  publisher = {IEEE Press}
}

@inproceedings{Val77,
  author    = {Leslie G. Valiant},
  title     = {Graph-theoretic arguments in low-level complexity},
  booktitle = {Mathematical Foundations of Computer Science 1977, 6th Symposium,
               Tatranska Lomnica, Czechoslovakia, September 5-9, 1977, Proceedings},
  pages     = {162--176},
  year      = {1977}
}

@article{Li23,
  title={Secret Sharing on Superconcentrator},
  author={Li, Yuan},
  journal={arXiv preprint arXiv:2302.04482},
  year={2023}
}

@inproceedings{Val75,
  title={On non-linear lower bounds in computational complexity},
  author={Valiant, Leslie G},
  booktitle={Proceedings of the seventh annual ACM symposium on Theory of computing},
  pages={45--53},
  year={1975}
}

@article{Val76,
  title={Graph-theoretic properties in computational complexity},
  author={Valiant, Leslie G},
  journal={Journal of Computer and System Sciences},
  volume={13},
  number={3},
  pages={278--285},
  year={1976},
  publisher={Academic Press}
}

@proceedings{DBLP:conf/stoc/STOC15,
  editor    = {David S. Johnson and
               Ronald Fagin and
               Michael L. Fredman and
               David Harel and
               Richard M. Karp and
               Nancy A. Lynch and
               Christos H. Papadimitriou and
               Ronald L. Rivest and
               Walter L. Ruzzo and
               Joel I. Seiferas},
  title     = {Proceedings of the 15th Annual {ACM} Symposium on Theory of Computing,
               25-27 April, 1983, Boston, Massachusetts, {USA}},
  publisher = {{ACM}},
  year      = {1983},
  bibsource = {dblp computer science bibliography, http://dblp.org}
}

@article{Tar75,
  title={Efficiency of a good but not linear set union algorithm},
  author={Tarjan, Robert Endre},
  journal={Journal of the ACM (JACM)},
  volume={22},
  number={2},
  pages={215--225},
  year={1975},
  publisher={ACM New York, NY, USA}
}
